\documentclass[10pt]{article}
\usepackage[left=1in,top=1in,right=1in,bottom=1in]{geometry}
\usepackage{times}

\usepackage{verbatim}
\usepackage{amsmath}
\usepackage{amsthm}
\usepackage{amssymb}
\usepackage{graphicx}

\usepackage{url}

\newtheorem{theorem}{Theorem}
\newtheorem{definition}{Definition}
\newtheorem{lemma}{Lemma}

\DeclareMathOperator*{\argmin}{arg\,min}

\begin{document}

\title{Human strategic decision making in parametrized games}

\author{
    Sam Ganzfried \\
    Ganzfried Research \\
    sam.ganzfried@gmail.com \\
}

\date{}

\maketitle

\begin{abstract}
  Many real-world games contain parameters which can affect payoffs, action spaces, and information states. For fixed values of the parameters, the game can be solved using standard algorithms. However, in many settings agents must act without knowing the values of the parameters that will be encountered in advance. Often the decisions must be made by a human under time and resource constraints, and it is unrealistic to assume that a human can solve the game in real time. We present a new framework that enables human decision makers to make fast decisions without the aid of real-time solvers. We demonstrate applicability to a variety of situations including settings with multiple players and imperfect information.
\end{abstract}

\section{Introduction}
Strong algorithms have been developed for game classes with many elements of complexity. For example, algorithms were recently able to defeat human professional players in 2-player~\cite{Moravcik17b:DeepStack,Brown17:Superhuman} and 6-player no-limit Texas hold 'em~\cite{Brown19:Superhuman}. These games have imperfect information, sequential actions, very large state spaces, and the latter has more than two players (solving multiplayer games is more challenging than two-player zero-sum games from a complexity-theoretic perspective). However, these algorithms all require an extremely large amount of computational resources for offline and/or online computations and for optimizing neural network hyperparameters. The algorithms also have a further limitation in that they are using all these resources just to solve for one very specific version of the game (e.g., Libratus and DeepStack assumed that all players start the hand with 200 times the big blind, and Pluribus assumed that all players start the hand with 100 times the big blind). In real poker, the stack sizes of the players will fluctuate as players win or lose hands, and will often differ from the standard starting values. While one could apply the same algorithms for any specific starting stack and blind values, there are too many possibilities to be able to run the algorithms on \emph{all} of them. 

In many real-world situations, there are some game parameters that will be encountered during gameplay that are unknown in advance (such as the stack sizes in poker). 
Furthermore, often the real-time decision maker will be a human, who may not have the ability to perform complex computations (even though such computations may have been performed offline in advance). For example, football teams have used statistical and game-theoretic models to decide whether or not to punt on a fourth down. In advance of game play, algorithms can be run that solve for optimal solutions in these models, perhaps by utilizing databases of historical play. However, in real time, the coach has only a matter of seconds (or minutes if a timeout is used) to make the decision. The optimal decision may depend on several factors that are not known until real time; for example, the score, how much time is remaining, the overall strength of the two teams, etc. The coach will need to weigh any offline algorithmic solution with newly-observed values of these parameters to make quick real-time decisions. Similar situations are also encountered frequently in national security. In recent years game-theoretic algorithms have been increasingly applied to national security domains. These algorithms must be applied offline with specific values of parameters used, which may differ from the actual values encountered in real time (e.g., an attacker may use more resources than anticipated). 

In these settings it is not sufficient to develop a successful algorithm for solving the game for one specific value of parameters; it is also necessary to develop a procedure for a decision maker to compute an optimal solution in real time for the newly-observed value of the parameters. Often the decision maker is a human, who may have limited technical expertise. It is not realistic to assume that the human decision maker can perform a complex algorithmic computation, tune or optimize a neural network, or perform a search over a large historical database during the seconds or minutes available to make the decision. While it is realistic to assume that a human can perform a table lookup to implement a previously computed strategy (even one stored in a large binary file), as described above this strategy may not apply to the observed parameter values. In addition, it is reasonable to assume that the human has access to a ``cheat sheet'' which contains a relatively small set of general rules to be applied depending on the parameter values encountered. It is also reasonable to assume that the decision maker can perform basic arithmetic (perhaps with the aid of a calculator). 

\section{Parametrized games}
We consider a setting where a decision maker must determine a strategy for a game $G_{\mathbf{\lambda}}$, which contains a vector of parameters $\mathbf{\lambda}$ drawn from some domain $\Lambda$ (typically subsets of the set of real numbers or integers). For any specific values of the parameters $\mathbf{\lambda}$, $G_{\mathbf{\lambda}}$ can be solved using standard algorithms. However, it is infeasible to solve the game in advance for all possible parameter values (there may be infinitely many). 

The parameters can affect different components of the game. They can affect the payoffs, set of players, strategy spaces, and/or private information. Often our goal will be to solve the game by computing a standard game-theoretic solution concept such as a Nash equilibrium. A second goal that we will also consider is opponent exploitation---computing a best response to perceived strategies of the opponents. The model of the opponents' strategies is determined in real time based on observations of the opponents' play, perhaps utilizing a prior distribution based on historical data. In this situation, the parameters can also affect the strategies of the opponents. In general our framework can allow for optimizing any objective, though we will focus on the natural objectives of computing Nash equilibrium or optimizing performance against assessments of opponents' strategies.

\begin{definition}
A \emph{strategic-form game} $G$ is a tuple $(N,S,u)$ with finite set of players $N = \{1,\ldots,n\}$, finite set of pure strategies $S_i$ for each player $i \in N$, and real-valued utility for each player for each strategy vector (aka \emph{strategy profile}), $u_i : \times_j S_j \rightarrow \mathbb{R}$.
\end{definition}

\begin{definition}
A \emph{payoff-parametrized strategic-form game} is a tuple $(\{G_{\mathbf{\lambda}}\},\Lambda)$ where for each real-valued vector of parameters $\mathbf{\lambda} \in \Lambda$, $G_{\mathbf{\lambda}}$ is a tuple $(N,S,u_{\mathbf{\lambda}})$ with finite set of players $N = \{1,\ldots,n\}$, finite set of pure strategies $S_i$ for each player $i \in N$, and real-valued utility for each player for each strategy vector (aka \emph{strategy profile}), $u_{\mathbf{\lambda},i} : \times_j S_j \rightarrow \mathbb{R}$.
\end{definition}

\begin{definition}
A \emph{strategy-parametrized strategic-form game} is a tuple $(\{G_{\mathbf{\lambda}}\},\Lambda)$ where for each real-valued vector of parameters $\mathbf{\lambda} \in \Lambda$, $G_{\mathbf{\lambda}}$ is a tuple $(N,S_{\mathbf{\lambda}},u)$ with finite set of players $N$, finite set of pure strategies $S_{\mathbf{\lambda},i}$ for each player $i \in N$, and real-valued utility for each player for each strategy vector (aka \emph{strategy profile}), $u_i : \times_j S_{\mathbf{\lambda},j} \rightarrow \mathbb{R}$.
\end{definition}

We can analogously define a \emph{player-parametrized game} where the set of players is determined by $\mathbf{\lambda}.$ We can also consider games that simultaneously consider several different types of parametrization. In general we will denote a strategic-form game that has payoff, strategy, and/or player parametrization as a \emph{parametrized strategic-form game}, $(\{G_{\mathbf{\lambda}}\},\Lambda)$, with $G_{\mathbf{\lambda}} = (\mathbf{\lambda},N_{\mathbf{\lambda}},S_{\mathbf{\lambda}},u_{\mathbf{\lambda}})$ for  $\mathbf{\lambda} \in \Lambda$.    

While the strategic form can be used to model simultaneous actions, another representation, called the \emph{extensive form}, is generally preferred when modeling settings that have sequential moves. The extensive form can also model simultaneous actions, as well as chance events and imperfect information (i.e., situations where some information is available to only some of the agents and not to others). Extensive-form games consist primarily of a game tree; each non-terminal node has an associated player (possibly \emph{chance}) that makes the decision at that node, and each terminal node has associated utilities for the players.  Additionally, game states are partitioned into \emph{information sets}, where the player whose turn it is to move cannot distinguish among the states in the same information set.  Therefore, in any given information set, a player must choose actions with the same distribution at each state contained in the information set.  
A pure strategy for player $i$ is a mapping that selects an action at each information set belonging to player $i$.

In typical imperfect-information extensive-form games, the initial move is a \emph{chance move} that assigns private information to each player (from a publicly-known distribution). For example, this could be private cards in poker, item valuations in auctions, or resource values in security games. Then players perform a sequence of publicly-observable actions until a terminal node is reached. Thus, for each sequence of public actions $p$, each player $i$ selects a strategy that is dependent on their own private information, $\tau_i$. Instead of viewing this decision as a selection of separate actions at each information set that follows the action sequence $p$, one can view this as a mapping that assigns an action for each possible value of the private information $\tau_i.$ If the set of possible values of the private information $\tau_i$ is dictated by a vector of parameters $\mathbf{\lambda}$, then we say that the game is \emph{information-parametrized}.


\section{Parametric decision lists}
Our goal for parametrized games is to develop a ``cheat sheet'' that allows a human decision maker to quickly select a strategy in real time for any possible value of the parameters $\mathbf{\lambda}.$ We propose a new structure, which we call a \emph{parametric decision list} (PDL), which contains a small set of rules that dictate which strategy should be played for every possible parameter value in a way that can be easily understood and implemented by a human. Similarly to a standard decision list, a parametric decision list consists of a series of conditions, each resulting in the output of a strategy. For game $G_{\mathbf{\lambda}}$, each condition will be of the form ``if $\mathbf{f_i}(\mathbf{\lambda}) \mathbf{o_i} \mathbf{0}$,'' where $\mathbf{f_i}$ is a vector-valued function of the parameters $\mathbf{\lambda}$, and $\mathbf{o_i}$ is vector of comparison operators from the set $\{<,\leq,>,\geq,=,\neq\}.$ For example if $\mathbf{\lambda} = (\lambda_1,\lambda_2)$, $\mathbf{f_i} = (4\lambda_1 + 2\lambda_2, 3\lambda_1 - 5\lambda_2)$, $\mathbf{o_i} = (\geq, <)$, then the condition would correspond to ``if $4\lambda_1 + 2\lambda_2 \geq 0$ and $3\lambda_1 - 5\lambda_2 < 0.$''  If the condition is satisfied, then (mixed) strategy $s_i$ is output. We can view the initial condition as corresponding to an ``If'' statement, subsequent conditions as ``Else if,'' and the final condition as ``Else.''   

\begin{definition}
\label{def:pdl}
A parametric decision list $L$ for $G_{\mathbf{\lambda}}$ is a tuple $L = (F, O, S)$, where $F = (\mathbf{f_i})$ is a sequence of functions $f_i : \mathbb{R}^{|\mathbf{\lambda}|} \rightarrow \mathbb{R}^w$, $O = (\mathbf{o_i})$ is a sequence of vectors of primitive comparison operators $\mathbf{o_i}$ with $|O| = |F|$, with $w = |\mathbf{o_i}|,$ and $S = (s_i)$ is a sequence of $|F| + 1$ (mixed) strategies. 
\end{definition}

We define the \emph{depth} of parametric decision list $L$ to equal the number of strategies, $|S| = |F| + 1.$ The first $|F|$ strategies correspond to when each of the $|F|$ conditions are met, and the final strategy corresponds to the default case when none are met (aka, the ``else'' condition). The \emph{width} of $L$ is equal to $w$, the length of the vectors $\mathbf{o_i}$. Each function outputs a $w$-dimensional vector $f_i$. Then each component $j$ is compared to 0 using operator $o_{ij}$. If all conditions of the operators are met, then the list dictates following strategy $s_i$. 

We say that parametrized game $G_{\mathbf{\lambda}}$ with objective function $g_{\mathbf{\lambda}}$ is $(d,w,\epsilon)$-implementable if there exists a parametric decision list $L$ with depth at most $d$ and width at most $w$ that achieves an objective value $g_{\mathbf{\lambda}}(s_L) \geq g^*_{\mathbf{\lambda}} - \epsilon$ for all $\mathbf{\lambda} \in \mathbf{\Lambda}$ for the strategy $s_L$ determined by $L$, where $g^*_{\mathbf{\lambda}}$ is the optimal value of the objective $g_{\mathbf{\lambda}}$ for $G_{\mathbf{\lambda}}$. The two primary objective functions we will be considering are the \emph{exploitability} of a strategy in a two-player zero-sum game, which is defined as the difference between the game value and payoff of the strategy against a best response to it, and performance of a strategy against a specific strategy (or distribution of strategies) for the opponents.



\section{Parameter sampling}
A second approach for generating a set of rules for a human decision-maker would be to repeatedly sample values for parameters $\mathbf{\lambda}_i$ and compute the optimal strategy $s_i$ in the parametrized game $G_{\mathbf{\lambda}_i}$ using standard approaches. Then when game $G_{\mathbf{\lambda}^*}$ is encountered in real time, we output the solution $s_{i}$ corresponding to the value of $i$ that minimizes $d(\mathbf{\lambda}_i,\mathbf{\lambda}^*)$, where $d$ is an appropriate distance metric. This sampling can be done uniformly at random over a suitably-chosen domain, or according to a more informative prior distribution if one is available. Assuming that the number of sampled games is relatively small and it is not too difficult to compute the distance function, this can potentially be another approach for human decision-making in parametrized games. 

Theorem~\ref{th:sampling} shows that as the number of samples grows large, this approach produces an optimal strategy if the payoffs are continuous functions of the parameters. The analysis is for the minimum exploitability metric in two-player zero-sum games, though similar analysis can also apply for other objectives. For simplicity we assume that the parameter $\lambda$ is one-dimensional and use the absolute value for the distance function, though an analogous result can be shown for the multi-dimensional case using an arbitrary distance metric.

\begin{lemma}
\label{le:cont-str}
Suppose all payoffs of $G''$ are within $\epsilon$ of the payoffs of $G'$.
Let $s'$ be a strategy profile in $G'.$
Then $|u^{G'}_i(s'_i,s'_{-i}) - u^{G''}_i(s'_i,s'_{-i})| \leq \epsilon.$
\end{lemma}
\begin{proof}
The utility against $s'_{-i}$ equals the sum of the utilities against each of the opponent's 
pure strategies $s_{-i}$ multiplied by the weight that $s'_{-i}$ places on $s_{-i}$. Since each of the
utilities in $G''$ is within $\epsilon$ of the utilities in $G'$, the weighted sums must be within $\epsilon$
of each other. 
\end{proof}

\begin{lemma}
\label{le:nemesis}
Suppose all payoffs of $G''$ are within $\epsilon$ of the payoffs of $G'$.
Let $s'_i$ be a strategy for player $i$ of $G'$, let $s^*_{-i}$ be a nemesis strategy against $s'_i$ in $G'$,
and let $s^{**}_{-i}$ be a nemesis strategy against $s'_i$ in $G''$.
Then $|u^{G'}_i(s'_i,s^*_{-i}) - u^{G''}_i(s'_i,s^{**}_{-i})| \leq \epsilon.$
\end{lemma}
\begin{proof}
Suppose that $u^{G'}_i(s'_i,s^*_{-i}) > u^{G''}_i(s'_i,s^{**}_{-i}) + \epsilon.$
By Lemma~\ref{le:cont-str}, 
$$u^{G'}_i(s'_i,s^{**}_{-i}) \leq u^{G''}_i(s'_i,s^{**}_{-i}) + \epsilon < u^{G'}_i(s'_i,s^*_{-i}),$$
which contradicts the fact that $s^*_i$ is a nemesis strategy against $s'_i$ in $G'.$
We obtain a similar contradiction if $$u^{G''}_i(s'_i,s^{**}_{-i}) > u^{G'}_i(s'_i,s^*_{-i}) + \epsilon.$$
So we conclude that $|u^{G'}_i(s'_i,s^*_{-i}) - u^{G''}_i(s'_i,s^{**}_{-i})| \leq \epsilon.$
\end{proof}

\begin{lemma}
\label{le:cont-values}
Suppose all payoffs of $G''$ are within $\epsilon$ of the payoffs of $G'$.
Let $v'_i$ be the game value of $G'$ to player $i$ and $v''_i$ be the value of $G''$.
Then $|v'_i - v''_i| \leq \epsilon.$
\end{lemma}
\begin{proof}
Let $s'_i$ be a Nash equilibrium strategy profile in $G'.$
Then $u_i(s'_i,s'_{-i}) = v'_i.$ 
Let $s^{**}_{-i}$ be a nemesis against $s'_i$ in $G''.$
By Lemma~\ref{le:nemesis}, 
$|u^{G'}_i(s'_i,s'_{-i}) - u^{G''}_i(s'_i,s^{**}_{-i})| \leq \epsilon.$
So $|v'_i - u^{G''}_i(s'_i,s^{**}_{-i})| \leq \epsilon.$
So $v'_i \leq u^{G''}_i(s'_i,s^{**}_{-i}) + \epsilon.$
We know that $u^{G''}_i(s'_i,s^{**}_{-i}) \leq v''_i.$
So $v'_i \leq v''_i + \epsilon.$
Similar reasoning shows that $v''_i \leq v'_i + \epsilon.$
So $|v'_i - v''_i| \leq \epsilon.$
\end{proof}

\begin{lemma}
\label{le:expl}
Suppose $s'$ is a Nash equilibrium of $G'$, and all payoffs of $G''$ are within $\epsilon$ of the payoffs of $G'$. Then the exploitability of $s'$ in $G''$ is at most $2\epsilon.$
\end{lemma}
\begin{proof}
Let $s'$ be a Nash equilibrium of $G'$, and $s^{**}_{-i}$ be a nemesis strategy to $s'$ in $G''.$ 
Let $v'_i$ denote the value of $G'$ and $v''_i$ denote the value of $G''.$
The exploitability of $s'_i$ in $G''$ equals $v''_i - u^{G''}_i(s'_i,s^{**}_{-i}).$
By the triangle inequality and Lemmas~\ref{le:cont-str} and~\ref{le:cont-values}, 
$$|v''_i - u^{G''}_i(s'_i,s^{**}_{-i})|$$ 
$$\leq |v''_i - v'_i| + |v'_i - u^{G'}_i(s'_i,s'_{-i})| + |u^{G'}_i(s'_i,s'_{-i}) - u^{G''}_i(s'_i,s^{**}_{-i})|$$
$$\leq \epsilon + 0 + \epsilon = 2\epsilon.$$
\end{proof}

Denote $u(s,\lambda)$ as $f_s(\lambda)$.  Without loss of generality suppose $\lambda \in [0,1]$ and that we repeatedly sample $\lambda_i$ from $U(0,1)$ and compute optimal strategy $s^*_i$ in the game $G_{\lambda_i}$. Then when we encounter $\lambda^*$ in real time, calculate $i^* = \argmin_{i} |\lambda^* - \lambda_i|$, and play
$s^*_{i^*}$. Suppose the game has $n = 2$ players, has $m$ actions per player, is zero sum, and that we take $t$ samples. Let $\epsilon_t$ denote the exploitability of $s^*_{i^*}$ in $G_{\lambda^*}$.

\begin{theorem}
\label{th:sampling}
If $f_s$ is continuous in $\lambda$ for all $s \in S$, then $\lim_{t \rightarrow \infty} E[\epsilon_t] = 0$.
\end{theorem}
\begin{proof}
Let $\epsilon > 0$ be arbitrary, and set $\epsilon' = \frac{\epsilon}{3}.$ From continuity of $f_s$, there exists $\delta_s > 0$ such that $|f_s(\lambda) - f_s(\lambda')| < \epsilon'$ for all $\lambda'$ such that $|\lambda' - \lambda| < \delta_s.$ 
Let $\delta = \min_{s \in S} \delta_s.$ 
Let $T = \frac{\ln\left(\frac{\epsilon}{6-2\epsilon}\right)}{\ln(1-2\delta)}$, and let $t \geq T$ be arbitrary.

Suppose that $\lambda^*$ is the actual value of $\lambda$ encountered. The probability that at least one of the sampled values $\lambda_i$ satisfies $|\lambda_i - \lambda^*| \leq \delta$ equals $1-(1-2 \delta)^t.$ If this occurs, then $\epsilon_t = 2 \epsilon' = \frac{2\epsilon}{3}$ by Lemma~\ref{le:expl}. Otherwise, the exploitability is at most 2 (since we assume all payoffs are in [0,1]). So
$$E[\epsilon_t] \leq \frac{2\epsilon}{3}\left(1-(1-2 \delta)^t\right) + 2 (1-2 \delta)^t$$
$$= \frac{2\epsilon}{3} + \frac{(6-2\epsilon)(1-2 \delta)^t}{3}$$
$$\leq \frac{2\epsilon}{3} + \frac{(6-2\epsilon)(1-2 \delta)^T}{3}$$
$$= \frac{2\epsilon}{3} + \frac{(6-2\epsilon)\left(\frac{\epsilon}{6-2\epsilon}\right)}{3}$$
$$= \frac{2\epsilon}{3} + \frac{\epsilon}{3} = \epsilon.$$
\end{proof}

\begin{theorem}
Suppose that we have $t$ samples $\lambda_i$ producing exploitability $\epsilon_t$. Then $G_{\mathbf{\lambda}}$
is $(t,t,\epsilon_t)$-implementable using the minimum exploitability objective.
\end{theorem}
\begin{proof}
We can construct parametric decision list $L$ as follows. First, we construct the function 
$f_1: \Lambda \rightarrow \mathbb{R}^k.$ 
The first component of $f_1$ is $|\lambda - \lambda_2| - |\lambda - \lambda_1|$ with operation $o_{11}$ $\leq$.
So in other words, this corresponds to the condition $|\lambda - \lambda_2| - |\lambda - \lambda_1| \leq 0.$
The $i$th component of $f_1$ is $|\lambda - \lambda_i| - |\lambda - \lambda_1|$, with $o_{1i}$ also $\leq.$
For the strategy, we set $s_i = s^*_i.$
These conditions put together tell us to play strategy $s_i$ if  $|\lambda - \lambda_1| \leq |\lambda - \lambda_i|$ for all $i > 1.$

In general the $i$th component of $f_j$ is $|\lambda - \lambda_i| - |\lambda - \lambda_j|$, with $o_{ji}$ $\leq$ and
$s_j = s^*_j.$ 
We can omit the final set of conditions for $f_t$ since it is implied by the first $t-1$ sets of conditions all failing, and output $s_t = s^*_t$ as the default strategy at depth $t.$
The constructed parametric decision list $L$ has depth and width $t$, and produces a strategy with exploitability of
$\epsilon_t$.
\end{proof}

\section{Comparison of approaches in 2 x 2 games}
In general it may be challenging to construct a small parametric decision list that achieves an approximately optimal value of the objective function. Similarly, for the sampling approach we may require a large number of samples to obtain a small approximation error. The sampling approach could be improved by first sampling as many values for the parameters as possible, then clustering the games generated (e.g., using k-means) into $k$ clusters. We then implement the strategy corresponding to the parameter values from the cluster mean with smallest distance from the parameters we encounter in real time. This approach would require an effective distance metric between parameter vectors that can be efficiently computed, while such a metric is not required for the parametric decision list approach. Furthermore, it may be challenging to determine the optimal value of $k$, and there is no guarantee that this approach with $k$ clusters will produce small error.  

In this section we compare the two approaches for the problem of computing Nash equilibrium strategies (i.e., using the objective of minimizing exploitability) in two-player zero-sum strategic-form games with two pure strategies per player. We can represent a two-player $2 \times 2$ game as a matrix $M$ depicted in Equation~\ref{eq:matrix}, where the parameters correspond to the payoffs of the row and column players. For general-sum games we can view this game as having 8 parameters ($a$--$h$), while for zero-sum games there are 4, since $b = -a, d = -c, f = -e, h = -g.$ While we can easily solve a specific game given the payoff parameters, we seek to construct a small set of rules that allow a human to easily obtain a solution for arbitrary parameter values.  




\begin{equation}
\label{eq:matrix}
M = 
\begin{bmatrix}
(a,b) & (c,d) \\
(e,f) & (g,h) \\
\end{bmatrix}
\end{equation}

We first explore the sampling approach. We generated 100,000 $2 \times 2$ 2-player zero-sum games with payoffs for the row player chosen uniformly at random in $[-1,1]$. We then used the first $k$ games from this set as training data, for various values of $k$. For each value of $k$, we generated 10,000 new games with uniform random payoffs. For each test game, we determine which of the $k$ training games is ``closest.'' For the distance metric, we use the L2-norm over vectors for the values $(a,c,e,g).$ We then compute the exploitability of the previously-computed equilibrium strategies from the closest training game in the new test game. The average exploitability over the 10,000 games is plotted as a function of $k$ in Figure~\ref{fi:2pzs}. From the figure we can see that as the number of training games gets large the average exploitability approaches zero, as expected. Surprisingly, training on just the first two games actually produces a very small average exploitability of 0.0129, while training on all 100,000 produces exploitability 0.0077. The exploitability for 3 training games (0.068) is significantly higher than that for 2, and jumps up sharply to a peak of 0.159 for 20 training games before descending towards zero. This erratic behavior shows that, on the one hand, the L2 distance metric has limitations for this problem and does not lead average exploitability to decrease monotonically with number of training games as we may expect. However, it also shows that it may be possible to generate a very small sample of training games (just 2) that produces a very low average exploitability.             

\begin{figure}[!ht]
\centering
\includegraphics[width=0.48\textwidth]{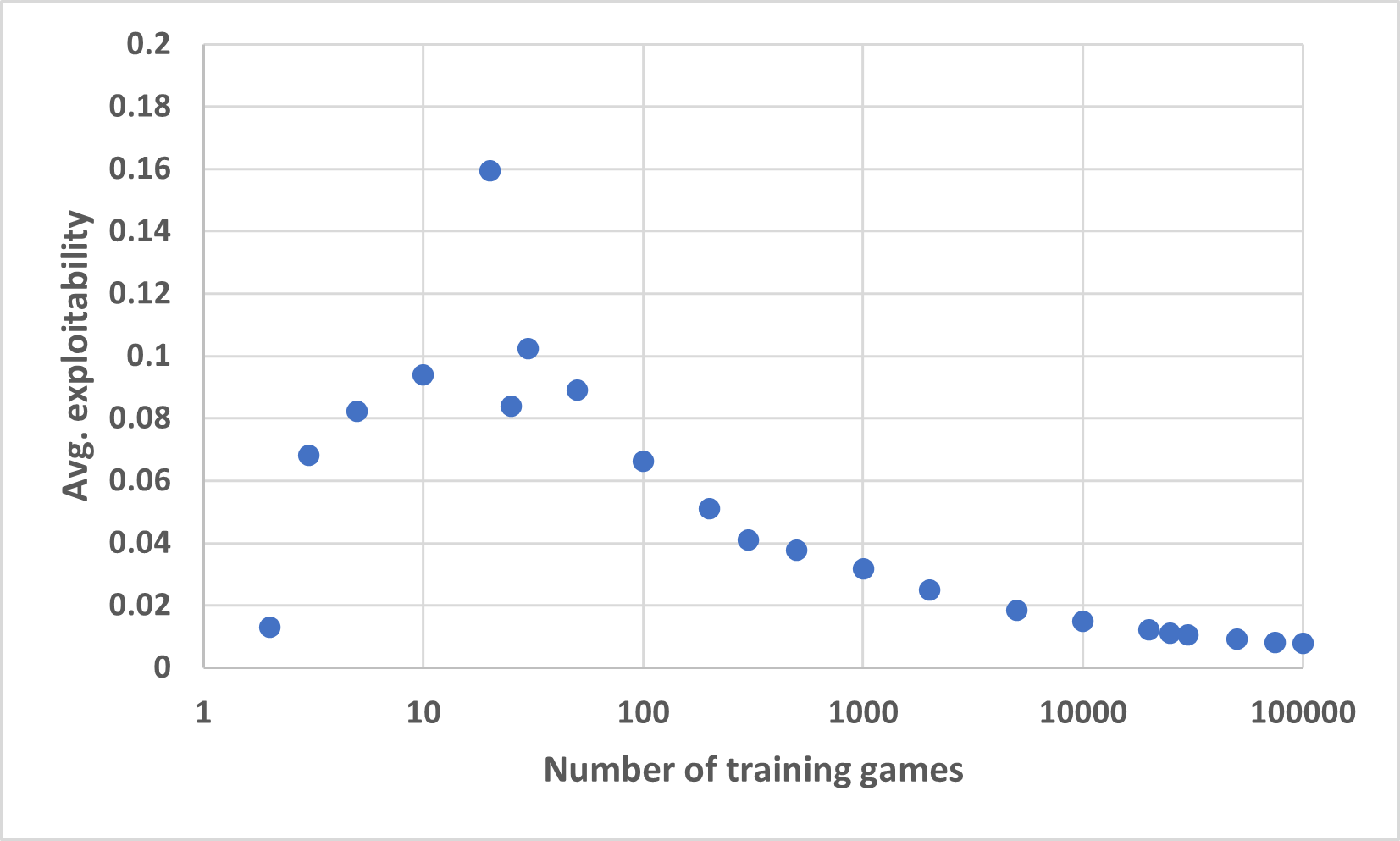}
\caption{Exploitability vs. number of training games for two-player zero-sum games with uniform-random payoffs in [-1,1], with results averaged over 10,000 games for each number of training games.}
\label{fi:2pzs}
\end{figure}

As it turns out, there actually exists a small parametric decision list (PDL) that computes an exact Nash equilibrium for $2 \times 2$ two-player general-sum strategic-form games (we can view zero-sum games as a special case). This is depicted below, using the notation for the parameters from Equation~\ref{eq:matrix}. We actually output equilibrium strategies for both players 1 and 2 in our PDL, though in practice we would only need to specify the strategy for the one player we are interested in. The final condition outputs a mixed strategy for both players, where the row player plays his strategies with probability $p$ and $1-p$ while the column player plays his strategies with probability $q$ and $1-q$. This PDL provides a $(4,2,0)$ implementation of the problem of computing a Nash equilibrium in this game class. A proof of correctness is in the appendix.   

\small
\begin{itemize}
\item If $a \geq e$ and $b \geq d$ then (1,1).
\item Else if $c \geq g$ and $d \geq b$ then (1,2).
\item Else if $e \geq a$ and $f \geq h$ then (2,1).
\item Else if $g \geq c$ and $h \geq f$ then (2,2).
\item Else $((p,1-p),(q,1-q))$ for $p = \frac{h-f}{b-f+h-d}$, $q = \frac{g-c}{a-c+g-e}$.
\end{itemize}
\normalsize

\section{Parametrized game examples}
In this section we present several examples of realistic games that depict various forms of our model.  In Section~\ref{se:jeopardy}, we present a game we call Simplified Final Jeopardy, which is a simplified two-player variant of the problem of determining how much to wager in final jeopardy (we can view it as the three-player version in which the third player has \$0). The three-player version is played for large amounts of money on the popular game show. We will assume that the player balances are fixed. Our model has 2 parameters which denote assessments of the probabilities that each player will answer correctly. We can assume that the assessments are based on observations of play throughout the game, as well as the category. These parameters affect the payoffs of the players. So the game exemplifies payoff parametrization. This game is two-player zero-sum, and we use the Nash equilibrium approximation/minimum exploitability objective.

In Section~\ref{se:kuhn} we consider a generalization of a simplified poker game that has been widely studied. Kuhn poker was one of the first games studied by game theorists~\cite{Kuhn50:Simplified}.  More recently, it has received significant attention in the artificial intelligence community as a tractable test problem for equilibrium-finding~\cite{Abou10:Using,Ganzfried10:Computing,Gordon05:No-Regret,Hawkin11:Automated,Koller95:Generating,Koller97:Representations} and opponent-exploitation~\cite{Hoehn05:Effective} algorithms. In the standard version, there are two players, each dealt one card from a three-card deck. We consider a variant in which the deck has $n$ cards. (Previously a version with a 13-card deck has been studied~\cite{Gordon13:One-card}.) The cards represent private information to the players, and therefore the game exemplifies information parametrization, though no strategy or payoff parametrization. 

Finally in Section~\ref{se:weakest-link}, we study a game model based on the game show Weakest Link. In the Weakest Link game show, eight contestants answer a series of trivia questions to accumulate a ``bank'' of money, with one contestant (the ``weakest link'') voted off at each round. When there are 2 contestants remaining, they face off for a series of five questions each, with the winner receiving the entire amount that was banked. In theory the champion could win ``up to a million dollars,'' though in practice the total bank ends up being in the \$50k--100k range. When three contestants remain, the players face an interesting strategic decision of deciding whom to vote for. Our model has five parameters: the total amount in the bank, assessments of the probability of winning against each opponent in the final round, and assessments of voting strategies of the opposing players (the probabilities that they will vote for each player). Thus, this is a three-player game using the opponent exploitation metric with strategy parametrization over the opponents' strategies as well as payoff parametrization.

These three games exemplify several of the different types of parametrization we have discussed: payoff, strategy, information, and opponent strategy. They exemplify the main objectives we have considered: minimizing exploitability for two-player zero-sum games, and maximizing opponent exploitation in multiplayer games. They also exemplify several different game classes (two-player zero-sum, imperfect-information, and multiplayer). These are simplified models of real popular games that are frequently played for large amounts of money. The purpose of these examples is to demonstrate the realistic applicability of the new concepts and frameworks we have presented. 

In each of the games human players must make their decision under extreme time pressure in real time without any computational assistance (though of course they can prepare a strategy in advance). While the optimal strategy can be computed easily in advance for fixed values of the parameters, the players must be prepared to face any possible values for the parameters. It turns out that for the games we consider we are able to construct parametric decision lists with small width and depth that exactly solve the problem based on the derivation of closed-form solutions. While in general larger games will clearly often not have closed-form solutions, the new concepts and frameworks can still be applied, though they may require the development of new focused algorithms. 


\subsection{Final Jeopardy}
\label{se:jeopardy}
In the simplified final jeopardy game, two players each have an amount $X_i$ and must select a non-negative amount $w_i \leq X_i$ to wager. We will assume that $X_1 = 5$, $X_2 = 3$, and the $w_i$ must be non-negative integers.  The player who finishes with a higher amount wins and obtains payoff 1, while the losing player obtains payoff 0 (we can then subtract 0.5 from each payoff to make the game zero sum). If there is a tie, then we assume each player obtains payoff 0.5. Finally, there are parameters $p_i$ that denote the probability that the players expect player $i$ to correctly answer the question. We assume that these values are correct and are common knowledge. 

For specific fixed values of the parameters $p_1,p_2$, the game is a two-player zero-sum strategic form game, and can be solved easily using standard algorithms. But such an approach is not helpful for a human player who must be prepared to be able to quickly construct his strategy in real time for any possible values of the parameters. The parametrized game is a $6 \times 4$ strategic-form game where the payoffs are functions of the parameters, and it is not obvious how to compute equilibrium strategies for all possible parameter values. As it turns out, we can construct the following small PDL which determines exact equilibrium strategies for all values of the parameters for player 1 (see appendix).

\begin{itemize}
\item If $p_2 = 0$ wager 0.
\item Else if $p_1 = 0$ wager 0.
\item Else if $p_1 = 1$ wager 2.
\item Else if $p_2 \geq \frac{1}{2}$ wager 2.
\item Otherwise wager 1 with probability $x = \frac{(1-p_1)(1-2p_2)}{1-p_1+p_1p_2}$ and wager 2 with probability $1-x.$
\end{itemize}

The PDL for player 2 is the following:
\begin{itemize}
\item If $p_2 = 0$ wager 0.
\item Else if $p_1 = 0$ wager 3.
\item Else if $p_1 = 1$ wager 0.
\item Else if $p_1 \geq \frac{1}{2}$ and $p_2 \geq \frac{1}{2}$ wager 2.
\item Else if $p_2 \geq \frac{1}{2}$ wager 3.
\item Otherwise wager 0 with probability $y = \frac{p_1 p_2}{1+p_1p_2 - p_1}$ and wager 3 with probability $1-y.$
\end{itemize}

\subsection{Generalized Kuhn Poker}
\label{se:kuhn}
The rules of three-card Kuhn poker are as follows:

\small
\begin{itemize}
\item Two players: $A$ and $B$
\item Both players ante \$1
\item Deck containing three cards: 1, 2, and 3
\item Each player is dealt one card uniformly at random
\item Player $A$ acts first and can either bet \$1 or check
\begin{itemize}
\item If $A$ bets, player $B$ can call or fold
\begin{itemize}
\item If $A$ bets and $B$ calls, then whoever has the higher card wins the \$4 pot
\item If $A$ bets and $B$ folds, then $A$ wins the entire \$3 pot
\end{itemize}
\item If $A$ checks, $B$ can bet \$1 or check.
\begin{itemize}
\item If $A$ checks and $B$ bets, then $A$ can call or fold.
\begin{itemize}
\item If $A$ checks, $B$ bets, and $A$ calls, then whoever has the higher card wins the \$4 pot
\item If $A$ checks, $B$ bets, and $A$ folds, B wins the \$3 pot
\end{itemize}
\item If $A$ checks and $B$ checks, then whoever has the higher card wins the \$2 pot
\end{itemize}
\end{itemize}
\end{itemize}
\normalsize

An analysis of the equilibria is provided in the appendix. The equilibrium strategies contain some elements of deceptive behavior, as is present in larger variants of poker.  For example, player $A$ sometimes checks with a 3 as a \emph{trap} or \emph{slowplay}, and both players sometimes bet with a 1 as a \emph{bluff}. 

Generalized Kuhn poker has the same rules as standard Kuhn poker except that the deck contains $n$ cards instead of 3.  We will denote the game with $n$ cards by $G_n.$  
We compute exact Nash equilibrium strategies in $G_n$ for the first time, which are presented in Appendix~\ref{ap:kuhn}. We show that they are almost entirely representable by a PDL, with the exception of a single boundary probability requiring a cumulative calculation.
While we do not present a theoretical proof of correctness, we empirically verified that this is an exact Nash equilibrium (up to $10^{-14}$ tolerance) for all $n \leq 500.$

\subsection{Weakest Link}
\label{se:weakest-link}
Our final example game is a model for the final voting round (3 players remain) in the show Weakest Link. Suppose that the total amount of money to be awarded to the winner is $W>0$ (the loser gets \$0). Suppose that if you are head-to-head against opponent 1 you will win with probability $p_1$, and against opponent 2 you will win with probability $p_2$. Assume that player 2 is stronger than player 1, so that $p_1 > p_2$. Finally, assume that player 1 will vote for you with probability $y_1$ (and therefore will vote for player 2 with probability 1-$y_1$), and that player 2 will vote for you with probability $y_2$. We will assume that clearly no player will vote for themselves. 

If there is a three-way tie, we will assume that each player is voted out with probability 1/3. (In reality, the ``statistically strongest link'' from the previous round gets to cast a tie-breaking vote, but we will ignore this aspect of the problem to simplify the analysis.) Under this assumption, our expected payoff in the case of a tie equals
$1/3 (W * p_1) + 1/3 (W * p_2) + 1/3 * 0 = W(p_1 + p_2)/3.$

Using this game model, our analysis (provided in the appendix) shows that we should vote for player 1 if
the following condition is met (and otherwise should vote for player 2):
$$2y_1p_2 + y_2p_2 + 3y_1y_2p_1 \geq 2y_2p_1 + y_1p_1 + 3y_1y_2p_2.$$
This constitutes an optimal depth 1 PDL for the objective.


\section{Related and future research}
While the majority of research on game-solving algorithms has assumed that all game parameters are fixed in advance and a computer agent will be implementing the strategies, some recent work has explored the construction of human-understandable strategy rules. One paper showed that equilibrium strategies for endgames in 2-player limit Texas hold 'em conform to one of three qualitative models, which enabled improved equilibrium computation algorithms~\cite{Ganzfried10:Computing}. Recent work in other imperfect-information poker games has applied machine learning algorithms (decision trees and regression) to compute human-understandable rules for fundamental situations: when a player should make a very large or small bet, and when a player should call a bet by the opponent~\cite{Ganzfried17b:Computing,Ganzfried20:Most}. There has also been prior study of parametrized games in the game theory literature~\cite{Page15:Parameterized,Flesch20:Parameterized}. 

While we have constructed optimal PDLs analytically for several example games, this may not be possible in general. In the future we plan to explore algorithms for computing small PDLs that achieve low objective error. Given the similarities to decision trees and decision lists, algorithms for computing those models may be useful for PDLs~\cite{Angelino18:Learning}. We would also like to perform experiments on human subjects to determine the practical usefulness of the PDL representation for human strategic decision-making in realistic complex games.

\bibliographystyle{plain}
\bibliography{C://FromBackup/Research/refs/dairefs}

\begin{thebibliography}{10}

\bibitem{Abou10:Using}
Nick {Abou Risk} and Duane Szafron.
\newblock Using counterfactual regret minimization to create competitive
  multiplayer poker agents.
\newblock In {\em {Proceedings of the International Conference on Autonomous
  Agents and Multi-Agent Systems (AAMAS)}}, 2010.

\bibitem{Angelino18:Learning}
Elaine Angelino, Nicholas Larus-Stone, Daniel Alabi, Margo Seltzer, and Cynthia
  Rudin.
\newblock Learning certifiably optimal rule lists for categorical data.
\newblock {\em Journal of Machine Learning Research}, 18(234):1--78, 2018.

\bibitem{Ankenman06:Mathematics}
Jerrod Ankenman and Bill Chen.
\newblock {\em The Mathematics of Poker}.
\newblock {ConJelCo LLC}, Pittsburgh, PA, USA, 2006.

\bibitem{Brown17:Superhuman}
Noam Brown and Tuomas Sandholm.
\newblock Superhuman {AI} for heads-up no-limit poker: Libratus beats top
  professionals.
\newblock {\em Science}, 359:418--424, 2017.

\bibitem{Brown19:Superhuman}
Noam Brown and Tuomas Sandholm.
\newblock Superhuman {AI} for multiplayer poker.
\newblock {\em Science}, 365:885--890, 2019.

\bibitem{Flesch20:Parameterized}
J\'{a}nos Flesch and Arkadi Predtetchinski.
\newblock {Parameterized games of perfect information}.
\newblock {\em Annals of Operations Research}, 287(2):683--699, April 2020.

\bibitem{Ganzfried20:Most}
Sam Ganzfried and Max Chiswick.
\newblock Most important fundamental rule of poker strategy.
\newblock In {\em Proceedings of the {F}lorida {A}rtificial {I}ntelligence
  {R}esearch {S}ociety {C}onference (FLAIRS)}, 2020.

\bibitem{Ganzfried10:Computing}
Sam Ganzfried and Tuomas Sandholm.
\newblock Computing equilibria by incorporating qualitative models.
\newblock In {\em {Proceedings of the International Conference on Autonomous
  Agents and Multi-Agent Systems (AAMAS)}}, 2010.

\bibitem{Ganzfried17b:Computing}
Sam Ganzfried and Farzana Yusuf.
\newblock Computing human-understandable strategies: {D}educing fundamental
  rules of poker strategy.
\newblock {\em Games}, 8(4):1--13, 2017.

\bibitem{Gordon13:One-card}
Geoff Gordon.
\newblock One-card poker.
\newblock \url{http://www.cs.cmu.edu/~ggordon/poker/}, 2013.

\bibitem{Gordon05:No-Regret}
Geoffrey~J. Gordon.
\newblock No-regret algorithms for structured prediction problems.
\newblock Technical Report CMU-CALD-05-112, Carnegie Mellon University, 2005.

\bibitem{Hawkin11:Automated}
John Hawkin, Robert Holte, and Duane Szafron.
\newblock Automated action abstraction of imperfect information extensive-form
  games.
\newblock In {\em {Proceedings of the AAAI Conference on Artificial
  Intelligence (AAAI)}}, pages 681--687, 2011.

\bibitem{Hoehn05:Effective}
Bret Hoehn, Finnegan Southey, Robert~C. Holte, and Valeriy Bulitko.
\newblock Effective short-term opponent exploitation in simplified poker.
\newblock In {\em Proceedings of the National Conference on Artificial
  Intelligence (AAAI)}, pages 783--788, 2005.

\bibitem{Koller95:Generating}
Daphne Koller and Avi Pfeffer.
\newblock Generating and solving imperfect information games.
\newblock In {\em Proceedings of the Fourteenth International Joint Conference
  on Artificial Intelligence (IJCAI)}, pages 1185--1192, Montreal, Canada,
  1995.

\bibitem{Koller97:Representations}
Daphne Koller and Avi Pfeffer.
\newblock Representations and solutions for game-theoretic problems.
\newblock {\em Artificial Intelligence}, 94(1):167--215, July 1997.

\bibitem{Kuhn50:Simplified}
H.~W. Kuhn.
\newblock Simplified two-person poker.
\newblock In H.~W. Kuhn and A.~W. Tucker, editors, {\em Contributions to the
  Theory of Games}, volume~1 of {\em Annals of Mathematics Studies, 24}, pages
  97--103. Princeton University Press, Princeton, New Jersey, 1950.

\bibitem{Moravcik17b:DeepStack}
Matej Morav{\v c}{\'\i}k, Martin Schmid, Neil Burch, Viliam Lis{\'y}, Dustin
  Morrill, Nolan Bard, Trevor Davis, Kevin Waugh, Michael Johanson, and Michael
  Bowling.
\newblock Deepstack: Expert-level artificial intelligence in heads-up no-limit
  poker.
\newblock {\em Science}, 356:508--513, 2017.

\bibitem{Page15:Parameterized}
Frank Page.
\newblock {Parameterized games, minimal Nash correspondences, and
  connectedness}.
\newblock LSE Research Online Documents on Economics 65102, London School of
  Economics and Political Science, LSE Library, September 2015.

\end{thebibliography}

\appendix
\section{Uniform random 2x2 two-player strategic-form games}
\label{ap:random} 
Consider the game defined by matrix $M$ in Equation~\ref{eq:matrix2}.
If the top-left cell is a pure-strategy equilibrium, then we must have
$a \geq e$ and $b \geq d$. The analysis is identical for the other pure
strategy profiles. Next, suppose there is a Nash equilibrium where one player
has support of size 2 and the other player has support of size 1.
Without loss of generality, suppose the column player's strategy has support 
of size 1 (Left) and the row player's strategy has support of size 2 (suppose it
puts probability $p$ on Top and $1-p$ on Bottom). Player 1 must be indifferent
between his two strategies, so $a = e.$ Player 2 cannot prefer R to L, so 
$$pb + (1-b)f \geq pd + (1-b) h.$$
If $b < d$ and $f < h$, then we would have
$$pb + (1-b)f < pd + (1-b) h,$$
which produces a contradiction. So we must have $b \geq d$ or $f \geq h$.
If $b \geq d$, then $(T,L)$ is a pure-strategy equilibrium, and if $f \geq h$
then $(B,L)$ is a pure-strategy equilibrium. So the existence of a Nash equilibrium
with support size 1 for one player and 2 for the other player implies the existence
of a pure-strategy Nash equilibrium. Therefore, if no pure-strategy Nash equilibrium
exists, then there must be a Nash equilibrium where both players' strategies have support
size 2. In this case it can be shown straightforwardly that the row player plays T with probability
$\frac{h-f}{b-f+h-d}$ and column player plays L with probability $\frac{g-c}{a-c+g-e}$.
This produces the PDL given below.
 
\begin{equation}
\label{eq:matrix2}
M = 
\begin{bmatrix}
(a,b) & (c,d) \\
(e,f) & (g,h) \\
\end{bmatrix}
\end{equation}

\begin{itemize}
\item If $a \geq e$ and $b \geq d$ then (1,1).
\item Else if $c \geq g$ and $d \geq b$ then (1,2).
\item Else if $e \geq a$ and $f \geq h$ then (2,1).
\item Else if $g \geq c$ and $h \geq f$ then (2,2).
\item Else $((p,1-p),(q,1-q))$ for $p = \frac{h-f}{b-f+h-d}$, $q = \frac{g-c}{a-c+g-e}$.
\end{itemize}

\section{Simplified Two-Player Final Jeopardy}
\label{ap:jeopardy}
In the two-player Final Jeopardy game, players have an amount $X_i$ and each player must select a non-negative amount $w_i \leq X_i$ to wager, where $X_i$ and $w_i$ are non-negative integers.  The player who finishes with a higher amount wins and obtains payoff 1, while the losing player obtains payoff 0 (we can then subtract 0.5 from each payoff to make the game zero sum). If there is a tie, then we assume each player obtains payoff 0.5. Finally, there are parameters $p_i$ that denote the probability that the players expect player $i$ to correctly answer the question. We assume that these values are correct and are common knowledge. 

In the simplified version we consider, the values $X_1 = 5$, $X_2 = 3$ are fixed. For specific fixed values of the parameters $p_1,p_2$, the game is two-player zero-sum strategic form game, and can be solved easily using standard algorithms. But such an approach is not helpful for a human player who must be prepared to be able to quickly construct his strategy in real time for any possible values of $p_1$ and $p_2$.  

\begin{itemize}
\item If player 1 wagers 0 and player 2 wagers 0, then player 1 wins with probability 1.
So player 1's expected payoff is $1 - 0.5 = 0.5.$
(Note that we are counting a win as having payoff 0.5, a loss as having payoff -0.5, and tie as having payoff 0, so that the game is zero sum.) 

\item If player 1 wagers 0 and player 2 wagers 1, then player 1 also wins with probability 1.
So player 1's expected payoff is $1 - 0.5 = 0.5.$

\item If player 1 wagers 0 and player 2 wagers 2, then player 1 wins with probability $1-p_2$, and the players tie with probability $p_2.$ So player 1's expected payoff is $1-p_2 + 0.5p_2 - 0.5 = 0.5 - 0.5p_2.$

\item If player 1 wagers 0 and player 2 wagers 3, then player 1 wins with probability $1-p_2$, and player 2 wins with probability $p_2.$ So player 1's expected payoff is $1-p_2 - 0.5 = 0.5 - p_2.$

\item If player 1 wagers 1 and player 2 wagers 0, then player 1 wins with probability 1.
So player 1's expected payoff is $1 - 0.5 = 0.5.$

\item If player 1 wagers 1 and player 2 wagers 1, then player 1 wins with probability $p_1 + (1-p_1)(1-p_2)$, and the players tie
with probability $(1-p_1)p_2.$
So player 1's expected payoff is 
$$p_1 + (1-p_1)(1-p_2) + 0.5 (1-p_1)p_2 - 0.5 = p_1 + 1 - p_1 - p_2 + p_1p_2 + 0.5p_2 - 0.5p_1p_2 - 0.5 = 0.5p_1p_2 - 0.5p_2 + 0.5.$$

\item If player 1 wagers 1 and player 2 wagers 2, then player 1 wins with probability $p_1 + (1-p_1)(1-p_2)$, and player 2 wins with
probability $(1-p_1)p_2.$
So player 1's expected payoff is 
$$p_1 + (1-p_1)(1-p_2) - 0.5 = p_1 + 1 - p_1 - p_2 - 0.5 = 0.5 - p_2.$$

\item If player 1 wagers 1 and player 2 wagers 3, then player 1 wins with probability $1-p_2$, player 2 wins with probability $(1-p_1)p_2$, and the players tie with probability $p_1 p_2.$
So player 1's expected payoff is
$$1-p_2 + 0.5 p_1 p_2 - 0.5 = 0.5 + 0.5p_1p_2 - p_2.$$

\item If player 1 wagers 2 and player 2 wagers 0, then player 1 wins with probability $p_1$, and the players tie with probability $1-p_1.$ So player 1's expected payoff is 
$$p_1 + 0.5(1-p_1) - 0.5 = p_1 + 0.5 - 0.5p_1 - 0.5 = 0.5p_1.$$

\item If player 1 wagers 2 and player 2 wagers 1, then player 1 wins with probability $p_1 + (1-p_1)(1-p_2)$, and player 2 wins with probability $(1-p_1)p_2.$
So player 1's expected payoff is 
$$p_1 + (1-p_1)(1-p_2) - 0.5 = p_1 + 1 - p_1 - p_2 + p_1p_2 - 0.5 = 0.5 - p_2 + p_1p_2.$$

\item If player 1 wagers 2 and player 2 wagers 2, then player 1 wins with probability $p_1 + (1-p_1)(1-p_2)$, and player 2 wins with probability $(1-p_1)p_2.$
So player 1's expected payoff is 
$$p_1 + (1-p_1)(1-p_2) - 0.5 = p_1 + 1 - p_1 - p_2 + p_1p_2 - 0.5 = 0.5 - p_2 + p_1p_2.$$

\item If player 1 wagers 2 and player 2 wagers 3, then player 1 wins with probability $p_1 + (1-p_1)(1-p_2)$, and player 2 wins with probability $(1-p_1)p_2.$
So player 1's expected payoff is 
$$p_1 + (1-p_1)(1-p_2) - 0.5 = p_1 + 1 - p_1 - p_2 + p_1p_2 - 0.5 = 0.5 - p_2 + p_1p_2.$$

\item If player 1 wagers 3 and player 2 wagers 0, then player 1 wins with probability $p_1$, and player 2 wins with probability $1-p_1.$ So player 1's expected payoff is $p_1 - 0.5.$

\item If player 1 wagers 3 and player 2 wagers 1, then player 1 wins with probability $p_1$, player 2 wins with probability $(1-p_1)p_2,$ and the players tie with probability $(1-p_1)(1-p_2).$
So player 1's expected payoff is 
$$p_1 + 0.5(1-p_1)(1-p_2) - 0.5 = p_1 + 0.5 - 0.5p_1 - 0.5p_2 + 0.5p_1p_2 - 0.5 = 0.5p_1p_2 + 0.5p_1 - 0.5p_2.$$

\item If player 1 wagers 3 and player 2 wagers 2, then player 1 wins with probability $p_1 + (1-p_1)(1-p_2)$, and player 2 wins with probability $(1-p_1)p_2.$
So player 1's expected payoff is 
$$p_1 + (1-p_1)(1-p_2) - 0.5 = p_1 + 1 - p_1 - p_2 + p_1p_2 - 0.5 = 0.5 - p_2 + p_1p_2.$$

\item If player 1 wagers 3 and player 2 wagers 3, then player 1 wins with probability $p_1 + (1-p_1)(1-p_2)$, and player 2 wins with probability $(1-p_1)p_2.$
So player 1's expected payoff is 
$$p_1 + (1-p_1)(1-p_2) - 0.5 = p_1 + 1 - p_1 - p_2 + p_1p_2 - 0.5 = 0.5 - p_2 + p_1p_2.$$

\item If player 1 wagers 4 and player 2 wagers 0, then player 1 wins with probability $p_1$, and player 2 wins with probability $1-p_1.$ So player 1's expected payoff is $p_1 - 0.5.$

\item If player 1 wagers 4 and player 2 wagers 1, then player 1 wins with probability $p_1$, and player 2 wins with probability $1-p_1.$ So player 1's expected payoff is $p_1 - 0.5.$

\item If player 1 wagers 4 and player 2 wagers 2, then player 1 wins with probability $p_1$, player 2 wins with probability $(1-p_1)p_2,$ and the players tie with probability $(1-p_1)(1-p_2).$
So player 1's expected payoff is 
$$p_1 + 0.5(1-p_1)(1-p_2) - 0.5 = p_1 + 0.5 - 0.5p_1 - 0.5p_2 + 0.5p_1p_2 - 0.5 = 0.5p_1p_2 + 0.5p_1 - 0.5p_2.$$

\item If player 1 wagers 4 and player 2 wagers 3, then player 1 wins with probability $p_1 + (1-p_1)(1-p_2)$, and player 2 wins with probability $(1-p_1)p_2.$
So player 1's expected payoff is 
$$p_1 + (1-p_1)(1-p_2) - 0.5 = p_1 + 1 - p_1 - p_2 + p_1p_2 - 0.5 = 0.5 - p_2 + p_1p_2.$$

\item If player 1 wagers 5 and player 2 wagers 0, then player 1 wins with probability $p_1$, and player 2 wins with probability $1-p_1.$ So player 1's expected payoff is $p_1 - 0.5.$

\item If player 1 wagers 5 and player 2 wagers 1, then player 1 wins with probability $p_1$, and player 2 wins with probability $1-p_1.$ So player 1's expected payoff is $p_1 - 0.5.$

\item If player 1 wagers 5 and player 2 wagers 2, then player 1 wins with probability $p_1$, and player 2 wins with probability $1-p_1.$ So player 1's expected payoff is $p_1 - 0.5.$

\item If player 1 wagers 5 and player 2 wagers 3, then player 1 wins with probability $p_1$, player 2 wins with probability $(1-p_1)p_2,$ and the players tie with probability $(1-p_1)(1-p_2).$
So player 1's expected payoff is 
$$p_1 + 0.5(1-p_1)(1-p_2) - 0.5 = p_1 + 0.5 - 0.5p_1 - 0.5p_2 + 0.5p_1p_2 - 0.5 = 0.5p_1p_2 + 0.5p_1 - 0.5p_2.$$
\end{itemize}

The game corresponds to the following payoff matrix, where the payoffs are for player 1 (we assume that player 1 is the row player and player 2 is the column player).

\begin{equation}
\begin{bmatrix}
0.5 &0.5 &0.5-0.5p_2 &0.5-p_2\\
0.5 &0.5p_1p_2 - 0.5p_2 + 0.5 &0.5-p_2 &0.5+0.5p_1p_2-p_2 \\
0.5p_1 &0.5 - p_2 + p_1p_2 &0.5 - p_2 + p_1p_2 &0.5 - p_2 + p_1p_2\\
p_1-0.5 &0.5p_1p_2 + 0.5p_1 - 0.5p_2 &0.5 - p_2 + p_1p_2 &0.5 - p_2 + p_1p_2\\
p_1-0.5 &p_1-0.5 &0.5p_1p_2 + 0.5p_1 - 0.5p_2 &0.5 - p_2 + p_1p_2\\
p_1-0.5 &p_1-0.5 &p_1-0.5 &0.5p_1p_2 + 0.5p_1 - 0.5p_2\\
\end{bmatrix}
\end{equation}

The payoff matrix for player 2 is the following (it is the same matrix with all payoffs negated):

\begin{equation}
\begin{bmatrix}
-0.5 &-0.5 &-0.5+0.5p_2 &-0.5+p_2\\
-0.5 &-0.5p_1p_2 + 0.5p_2 - 0.5 &-0.5+p_2 &-0.5-0.5p_1p_2+p_2 \\
-0.5p_1 &-0.5 + p_2 - p_1p_2 &-0.5 + p_2 - p_1p_2 &-0.5 + p_2 - p_1p_2\\
-p_1+0.5 &-0.5p_1p_2 - 0.5p_1 + 0.5p_2 &-0.5 + p_2 - p_1p_2 &-0.5 + p_2 - p_1p_2\\
-p_1+0.5 &-p_1+0.5 &-0.5p_1p_2 - 0.5p_1 + 0.5p_2 &-0.5 + p_2 - p_1p_2\\
-p_1+0.5 &-p_1+0.5 &-p_1+0.5 &-0.5p_1p_2 - 0.5p_1 + 0.5p_2\\
\end{bmatrix}
\end{equation}

\begin{itemize}
\item (0,0) is a Nash equilibrium if $p_2 = 0.$
\item Else (0,3) is a Nash equilibrium if $p_1 = 0.$
\item Else (2,0) is a Nash equilibrium if $p_1 = 1.$
\item Else (2,2) is a Nash equilibrium if $p_1 \geq \frac{1}{2}, p_2 \geq \frac{1}{2}.$
\item Else (2,3) is a Nash equilibrium if $p_1 < \frac{1}{2}, p_2 \geq \frac{1}{2}.$
\item Else P1 wagers 1 with probability $x = \frac{(1-p_1)(1-2p_2)}{1-p_1+p_1p_2}$ and wagers 2 with probability $1-x$,
and P2 wagers 0 with probability $y = \frac{p_1 p_2}{1+p_1p_2 - p_1}$ and 3 with probability $1-y$ is a Nash equilibrium if $p_2 < \frac{1}{2}.$
\end{itemize}

Expected payoff for player 1 against this strategy is:
$$0.5 y + (1-y)(0.5 + 0.5p_1p_2 - p_2)$$
$$= 0.5 y + 0.5 + 0.5p_1p_2 - p_2 - 0.5y - 0.5y p_1p_2 + yp_2$$
$$= 0.5 + 0.5 p_1 p_2 - p_2 - 0.5yp_1p_2 + yp_2$$
$$= 0.5 + 0.5p_1p_2 - p_2 + y(p_2 - 0.5 p_1p_2)$$
$$= 0.5 + 0.5p_1p_2 - p_2 + \frac{p_1 p_2(p_2 - 0.5 p_1p_2)}{1+p_1p_2 - p_1}$$
$$= \frac{(0.5 + 0.5p_1p_2 - p_2)(1+p_1p_2 - p_1) + p_1 p_2(p_2 - 0.5 p_1p_2)}{1+p_1p_2 - p_1}$$
$$= \frac{0.5 + 0.5p_1p_2 - 0.5p_1 + 0.5p_1p_2 + 0.5p^2_1p^2_2 - 0.5p^2_1p_2 - p_2 - p_1p^2_2 + p_1p_2+ p_1p^2_2 - 0.5p^2_1p^2_2}{1+p_1p_2 - p_1}$$
$$= \frac{0.5 + 2p_1p_2 - 0.5p_1 - 0.5p^2_1p_2 - p_2}{1+p_1p_2 - p_1}$$

Expected payoff for player 1 wagering 0 against this strategy:
$$0.5y + (1-y)(0.5-p_2) = 0.5y + 0.5 - p_2 + 0.5y + yp_2 = 0.5 - p_2 + y(1+p_2)$$
Suppose this exceeds the payoff of playing the above strategy. Then
$$0.5 - p_2 + y(1+p_2) > 0.5 + 0.5p_1p_2 - p_2 + y(p_2 - 0.5 p_1p_2)$$
$$y + yp_2 > 0.5p_1p_2 + yp_2 - 0.5 p_1 p^2_2$$
$$y > 0.5p_1p_2 -  0.5 p_1 p^2_2$$
$$\frac{p_1 p_2}{1+p_1p_2 - p_1} > 0.5p_1p_2 -  0.5 p_1 p^2_2$$
$$p_1p_2 > (0.5p_1p_2 -  0.5 p_1 p^2_2)(1+p_1p_2 - p_1)$$
$$p_1p_2 > 0.5p_1p_2 + 0.5p^2_1p^2_2 - 0.5p^2_1p_2 - 0.5 p_1 p^2_2 - 0.5 p^2_1p^3_2 + 0.5p^2_1p^2_2$$
$$0 > -0.5p_1p_2 + p^2_1p^2_2 - 0.5p^2_1p_2 - 0.5 p_1 p^2_2 - 0.5 p^2_1p^3_2$$
$$0 > p_1p_2 + 2p^2_1p^2_2 - p^2_1p_2 - p_1 p^2_2 - p^2_1p^3_2$$
$$0 > 1 + 2p_1p_2 - p_1 - p_2 - p^2_2$$
Which always false for $p_2 < \frac{1}{2}.$

Expected payoff for player 1 wagering 3 against this strategy:
$$y(p_1 - 0.5) + (1-y)(0.5-p_2+p_1p_2)$$
$$= yp_1 - 0.5y + 0.5 - p_2 + p_1p_2 - 0.5y + yp_2 - yp_1p_2$$
$$= yp_1 - y + 0.5 - p_2 + p_1p_2 + yp_2 - yp_1p_2$$
$$= 0.5 - p_2 + p_1p_2 + y(p_1 - 1 + p_2 - p_1p_2)$$
Suppose this exceeds the payoff of playing the above strategy. Then
$$0.5 - p_2 + p_1p_2 + y(p_1 - 1 + p_2 - p_1p_2) > 0.5 + 0.5p_1p_2 - p_2 + y(p_2 - 0.5 p_1p_2)$$
$$0.5p_1p_2 + y(p_1 - 1 - 0.5p_1p_2) > 0$$
$$0.5p_1p_2 + \frac{p_1 p_2(p_1 - 1 - 0.5p_1p_2)}{1+p_1p_2 - p_1} > 0$$
$$0.5p_1p_2(1+p_1p_2 - p_1) + p_1 p_2(p_1 - 1 - 0.5p_1p_2) > 0$$
$$0.5p_1p_2 + 0.5p^2_1p^2_2 - 0.5p^2_1p_2 + p^2_1p_2 - p_1p_2 - 0.5p^2_1p^2_2 > 0$$
$$-0.5p_1p_2  + 0.5p^2_1p_2 > 0$$
$$-0.5 + 0.5 p_1 > 0$$
$$p_1 > 1$$
Which is a contradiction.

Expected payoff for player 1 wagering 4 against this strategy is identical to the expected payoff of wagering 3, so the same argument will apply.

Expected payoff for player 1 wagering 5 against this strategy:
$$y(p_1 - 0.5) + (1-y)(-0.5p_1p_2 - 0.5p_1 + 0.5p_2)$$
$$= yp_1 - 0.5y -0.5p_1p_2 - 0.5p_1 + 0.5p_2 + 0.5yp_1p_2 + 0.5yp_1 - 0.5yp_2$$ 
$$= 1.5 yp_1 - 0.5y -0.5p_1p_2 - 0.5p_1 + 0.5p_2 + 0.5yp_1p_2 - 0.5yp_2$$
$$= -0.5p_1p_2 - 0.5p_1 + y(1.5p_1 - 0.5 + 0.5p_1p_2 - 0.5p_2)$$
Suppose this exceeds the payoff of playing the above strategy. Then
$$= -0.5p_1p_2 - 0.5p_1 + y(1.5p_1 - 0.5 + 0.5p_1p_2 - 0.5p_2) > 0.5 + 0.5p_1p_2 - p_2 + y(p_2 - 0.5 p_1p_2)$$
$$-0.5 - p_1p_2 - 0.5p_1 + p_2 + y(1.5p_1 - 0.5 + p_1p_2 - 1.5p_2) > 0$$
$$-0.5 - p_1p_2 - 0.5p_1 + p_2 + \frac{p_1 p_2(1.5p_1 - 0.5 + p_1p_2 - 1.5p_2)}{1+p_1p_2 - p_1} > 0$$
$$(-0.5 - p_1p_2 - 0.5p_1 + p_2)(1+p_1p_2 - p_1) + p_1 p_2(1.5p_1 - 0.5 + p_1p_2 - 1.5p_2) > 0$$
$$-0.5 - p_1p_2 - 0.5p_1 + p_2 -0.5p_1p_2 - p^2_1p^2_2 - 0.5p^2_1p_2 + p_1p^2_2 
+0.5p_1 + p^2_1p_2 + 0.5p^2_1 - p_1p_2 + 1.5p^2_1p_2 - 0.5p_1p_2 + p^2_1p^2_2 - 1.5p_1p^2_2 > 0$$
$$-0.5 - 3p_1p_2 + p_2 + 2p^2_1p_2 -0.5p_1p^2_2  + 0.5p^2_1 > 0$$
Derivative with respect to $p_2$ and setting to 0 gives
$$p_2 = \frac{1 - 3p_1+2p^2_1}{p_1}$$
For this to be between 0 and $\frac{1}{2}$, we must have $0.35 \leq p_1 \leq 0.5.$

$$-0.5 - 3p_1\frac{1 - 3p_1+2p^2_1}{p_1} + \frac{1 - 3p_1+2p^2_1}{p_1} + 2p^2_1\frac{1 - 3p_1+2p^2_1}{p_1} -0.5p_1\left(\frac{1 - 3p_1+2p^2_1}{p_1}\right)^2  + 0.5p^2_1 > 0$$
$$-0.5p_1 -3p_1(1 - 3p_1+2p^2_1) + (1 - 3p_1+2p^2_1) + 2p^2_1(1 - 3p_1+2p^2_1) 
-0.5(1 - 3p_1+2p^2_1)^2 + 0.5p^3_1 > 0$$
$$-0.5p_1 - 3p_1 + 9p^2_1 - 6p^3_1 + 1 - 3p_1+2p^2_1 + 2p^2_1 - 6p^3_1 + 4p^4_1 
- 0.5(1 + 9p^2_1 + 4p^4_1 -6p_1 + 4p^2_1 - 12p^3_1) + 0.5p^3_1 > 0$$
$$-0.5p_1 - 3p_1 + 9p^2_1 - 6p^3_1 + 1 - 3p_1+2p^2_1 + 2p^2_1 - 6p^3_1 + 4p^4_1 
- 0.5 - 4.5p^2_1 - 2p^4_1 + 3p_1 -2p^2_1 + 6p^3_1 + 0.5p^3_1 > 0$$
$$0.5 -3.5p_1 + 6.5p^2_1 - 5.5p^3_1 + 2p^4_1 > 0$$
The LHS is always negative for $0.35 \leq p_1 \leq 0.5.$
So we have a contradiction, and we have shown that player 1 can't profitably deviate.

Expected payoff for player 2 against the strategy of player 1:
$$-0.5p_1 + x(0.5p_1 - 0.5)$$
$$= \frac{-0.5 - 2p_1p_2 + 0.5p_1 + 0.5p^2_1p_2 + p_2}{1+p_1p_2 - p_1}$$

Expected payoff for player 2 wagering 1 against this strategy:
$$x(-0.5p_1p_2+0.5p_2-0.5) + (1-x)(-0.5+p_2-p_1p_2)$$
$$= -0.5xp_1p_2 + 0.5xp_2 - 0.5x - 0.5 + p_2 - p_1p_2 + 0.5x - xp_2 + xp_1p_2$$
$$= 0.5xp_1p_2 - 0.5xp_2 - 0.5 + p_2 - p_1p_2$$
$$= -0.5 + p_2 - p_1p_2 + x(0.5p_1p_2 - 0.5p_2)$$
Suppose this exceeds the payoff of playing the above strategy. Then
$$= -0.5 + p_2 - p_1p_2 + x(0.5p_1p_2 - 0.5p_2) > -0.5p_1 + x(0.5p_1 - 0.5)$$
$$-0.5 + p_2 - p_1p_2 + 0.5p_1 + x(0.5p_1p_2 - 0.5p_2 - 0.5p_1 + 0.5) > 0$$
$$(-0.5 + p_2 - p_1p_2 + 0.5p_1)(1-p_1+p_1p_2) + (1-p_1)(1-2p_2)(0.5p_1p_2 - 0.5p_2 - 0.5p_1 + 0.5) > 0$$
$$-0.5 + 0.5p_1 - 0.5p_1p_2 + p_2 - p_1p_2 + p_1p^2_2 - p_1p_2 + p^2_1p_2 - p^2_1p^2_2 + 0.5p_1 - 0.5p^2_1 + 0.5p^2_1p_2$$
$$+0.5p_1p_2 - 0.5p_2 - 0.5p_1 + 0.5 - 0.5p^2_1p_2 + 0.5p_1p_2 + 0.5p^2_1 - 0.5p_1 - p_1p^2_2 + p^2_2 + p_1p_2 - p_2 > 0$$
$$-0.5p_1p_2 - 0.5p_2 + p^2_1p_2 - p^2_1p^2_2 + p^2_2> 0$$
$$-0.5p_1 - 0.5 + p^2_1 - p^2_1p_2 + p_2> 0$$
This is never true for $p_2 < \frac{1}{2}.$ So we have a contradiction.

Expected payoff for player 2 wagering 2 against this strategy:
$$x(-0.5+p_2) + (1-x)(-0.5+p_2-p_1p_2)$$
$$= -0.5x + xp_2 - 0.5 + p_2 - p_1p_2 + 0.5x - xp_2 + xp_1p_2$$
$$= -0.5 + p_2 - p_1p_2 + xp_1p_2$$
Suppose this exceeds the payoff of playing the above strategy. Then
$$(-0.5 + p_2 - p_1p_2)(1-p_1+p_1p_2) + p_1p_2(1-p_1)(1-2p_2) > 0$$
$$-0.5 + p_2 - p_1p_2 + 0.5p_1 - p_1p_2 + p^2_1p_2 -0.5p_1p_2 + p_1p^2_2 - p^2_1p^2_2 
+ p_1p_2 - p^2_1p_2 - 2p_1p^2_2 + 2p^2_1p^2_2 > 0$$
$$-0.5 + p_2 - 1.5p_1p_2 + 0.5p_1 - p_1p^2_2 + p^2_1p^2_2 > 0$$
Taking derivative wrt $p_2$ and setting to 0 gives
$$p_2 = \frac{2-3p_1}{4p_1-4p^2_1}$$
To obtain $0 < p_2 < \frac{1}{2}$, we must have $0.5 < p_1 < \frac{2}{3}.$
$$-0.5 + \frac{2-3p_1}{4p_1-4p^2_1} - 1.5p_1 \frac{2-3p_1}{4p_1-4p^2_1} + 0.5p_1 - p_1 
\left(\frac{2-3p_1}{4p_1-4p^2_1}\right)^2 + p^2_1 \left(\frac{2-3p_1}{4p_1-4p^2_1}\right)^2 > 0$$
$$-0.5(4p_1-4p^2_1)^2 + (2-3p_1)(4p_1-4p^2_1) - 1.5p_1 (2-3p_1)(4p_1-4p^2_1) + 0.5p_1(4p_1-4p^2_1)^2
- p_1(2-3p_1)^2 + p^2_1(2-3p_1)^2 > 0$$
$$-0.5p_1(4-4p_1)^2 + (2-3p_1)(4-4p_1) - 1.5p_1 (2-3p_1)(4-4p_1) + 0.5p^2_1(4-4p_1)^2
- (2-3p_1)^2 + p_1(2-3p_1)^2 > 0$$
$$-8p_1 + 16 p^2_1 -8p^3_1 + 8 -8p_1 - 12 p_1 + 12p^2_1 - 12p_1 + 12p^2_1 + 18 p^2_1 - 18 p^3_1 + 8 p^2_1 - 16 p^3_1 + 8 p^4_1  - 4 + 12p_1 - 9p^2_1 + 4p_1 - 12p^2_1 + 9p^3_1 > 0$$
$$8x^4 - 33x^3 + 45x^2 -24x + 4 > 0$$
The LHS is always negative for $0.5 < p_1 < \frac{2}{3}.$
So we have a contradiction, and we have shown that player 2 can't profitably deviate.

\section{Generalized Kuhn Poker}
\label{ap:kuhn}
Kuhn poker was one of the first games studied by game theorists and was developed by Harold Kuhn in 1950~\cite{Kuhn50:Simplified}.  More recently, it has received significant attention in the artificial intelligence community as a tractable test problem for equilibrium-finding~\cite{Abou10:Using,Ganzfried10:Computing,Gordon05:No-Regret,Hawkin11:Automated,Koller95:Generating,Koller97:Representations} and opponent-exploitation~\cite{Hoehn05:Effective} algorithms. In the standard version, there are two players, each dealt one card from a three-card deck. We consider a variant in which the deck has $n$ cards. (Previously a version with a 13-card deck has been studied~\cite{Gordon13:One-card}.)

\begin{itemize}
\item Two players: $A$ and $B$
\item Both players ante \$1
\item Deck containing three cards: 1, 2, and 3
\item Each player is dealt one card uniformly at random
\item Player $A$ acts first and can either bet \$1 or check
\begin{itemize}
\item If $A$ bets, player $B$ can call or fold
\begin{itemize}
\item If $A$ bets and $B$ calls, then whoever has the higher card wins the \$4 pot
\item If $A$ bets and $B$ folds, then $A$ wins the entire \$3 pot
\end{itemize}
\item If $A$ checks, $B$ can bet \$1 or check.
\begin{itemize}
\item If $A$ checks and $B$ bets, then $A$ can call or fold.
\begin{itemize}
\item If $A$ checks, $B$ bets, and $A$ calls, then whoever has the higher card wins the \$4 pot
\item If $A$ checks, $B$ bets, and $A$ folds, then B wins the \$3 pot
\end{itemize}
\item If $A$ checks and $B$ checks, then whoever has the higher card wins the \$2 pot
\end{itemize}
\end{itemize}
\end{itemize}

It is known that for any $0 \leq \alpha \leq 1$ the following strategy profile is an equilibrium (and that these are all the equilibria) ~\cite{Kuhn50:Simplified}.  

\begin{itemize}
\item $A$ bets with a 1 in the first round with probability $\frac{\alpha}{3}$
\item $A$ always checks with a 2 in the first round 
\item $A$ bets with a 3 in the first round with probability $\alpha$
\item If $A$ bets in the first round, then:
\begin{itemize}
\item $B$ always folds with a 1
\item $B$ calls with a 2 with probability $\frac{1}{3}$
\item $B$ always calls with a 3
\end{itemize}
\item If $A$ checks in the first round, then:
\begin{itemize}
\item $B$ bets with a 1 with probability $\frac{1}{3}$
\item $B$ always checks with a 2
\item $B$ always bets with a 3
\end{itemize}
\item If $A$ checks and $B$ bets, then:
\begin{itemize}
\item $A$ always folds with a 1
\item $A$ calls with a 2 with probability $\frac{\alpha}{3} + \frac{1}{3}$
\item $A$ always calls with a 3
\end{itemize}
\end{itemize}

Several immediate observations can be made from the equilibria of three-card Kuhn poker.
\begin{enumerate}
\item There are infinitely many equilibria
\item There are no pure strategy equilibria
\item Equilibrium strategies contain some elements of deceptive behavior.  For example, player $A$ sometimes checks with a 3 as a \emph{trap} or \emph{slowplay}, and both players sometimes bet with a 1 as a \emph{bluff}.
\end{enumerate}

Generalized Kuhn poker (GKP) has the same rules as standard Kuhn poker except that the deck contains $n$ cards instead of 3.  We will denote the game with $n$ cards by $G_n.$  Unlike the $n = 3$ case, no closed-form solution has previously been derived for general $n.$ We compute the solution to this game for the first time. Before presenting the solution, we present the solution for the continuous version where cards are dealt uniformly on the continuous interval $[0,1]$ where larger values correspond to stronger hands. A closed-form solution has been computed for this game~\cite{Ankenman06:Mathematics}, which can be viewed as the limit of $G_n$ as $n$ approaches infinity. 

Let hands $x,y \in [0,1]$, where larger values correspond to stronger hands.

\begin{align*}
\sigma_1^{\mathrm{init}}(\text{bet} \mid x)
&=
\begin{cases}
1, & 0 \le x \le \frac{1}{9},\\
0, & \frac{1}{9} < x < \frac{2}{3},\\
1, & \frac{2}{3} \le x \le 1,
\end{cases}
\\[1em]
\sigma_1^{\mathrm{cb}}(\text{call} \mid x)
&=
\begin{cases}
0, & 0 \le x < \frac{1}{3},\\
1, & \frac{1}{3} \le x \le 1,
\end{cases}
\\[1em]
\sigma_2^{\mathrm{bet}}(\text{call} \mid y)
&=
\begin{cases}
0, & 0 \le y < \frac{1}{3},\\
1, & \frac{1}{3} \le y \le 1,
\end{cases}
\\[1em]
\sigma_2^{\mathrm{check}}(\text{bet} \mid y)
&=
\begin{cases}
1, & 0 \le y \le \frac{1}{6},\\
0, & \frac{1}{6} < y < \frac{1}{2},\\
1, & \frac{1}{2} \le y \le 1.
\end{cases}
\end{align*}

Now we present an exact Nash equilibrium for the discrete version of $n$-card Kuhn poker, for $n \geq 3.$
Let $i \in \{1,\dots,n\}$ denote the card index, with larger cards stronger.
We denote the first player to act as player 1, and the second as player 2. In the strategies below, the subscript denotes the player, the superscript denotes the betting history, and the argument in parentheses is the specified action given the card $i.$ While we do not present a proof of correctness, we empirically verified that for $3 \leq n \leq 500$ the exploitabilities of the strategies for each player are less than $10^{-14}$, which is essentially 0.

\subsection*{Player 2 (all $n$)}

\begin{align*}
\sigma_2^{\mathrm{bet}}(\text{call} \mid i)
&=
\begin{cases}
0,
& i \le \left\lfloor \dfrac{n-1}{3} \right\rfloor + 1,
\\[6pt]
1-\left(\dfrac{n-1}{3}-\left\lfloor \dfrac{n-1}{3}\right\rfloor\right),
& i = \left\lfloor \dfrac{n-1}{3} \right\rfloor + 2,
\\[6pt]
1,
& i \ge \left\lfloor \dfrac{n-1}{3} \right\rfloor + 3,
\end{cases}
\\[1em]
\sigma_2^{\mathrm{check}}(\text{bet} \mid i)
&=
\begin{cases}
1,
& i \le \left\lfloor \dfrac{n-1}{6} \right\rfloor,
\\[6pt]
\dfrac{n-1}{6}-\left\lfloor \dfrac{n-1}{6}\right\rfloor,
& i = \left\lfloor \dfrac{n-1}{6} \right\rfloor + 1,
\\[6pt]
0,
& \left\lfloor \dfrac{n-1}{6} \right\rfloor + 2
\le i \le
\left\lfloor \dfrac{n-1}{2} \right\rfloor + 1,
\\[6pt]
1-\left(\dfrac{n-1}{2}-\left\lfloor \dfrac{n-1}{2}\right\rfloor\right),
& i = \left\lfloor \dfrac{n-1}{2} \right\rfloor + 2,
\\[6pt]
1,
& i \ge \left\lfloor \dfrac{n-1}{2} \right\rfloor + 3.
\end{cases}
\end{align*}

\subsection*{Player 1 initial action}

\paragraph{Odd $n$:}
\begin{align*}
\sigma_1^{\mathrm{init}}(\text{bet} \mid i)
&=
\begin{cases}
1,
& i \le \left\lfloor \dfrac{n-1}{9} \right\rfloor,
\\[6pt]
\dfrac{n-1}{9}-\left\lfloor \dfrac{n-1}{9}\right\rfloor,
& i = \left\lfloor \dfrac{n-1}{9} \right\rfloor + 1,
\\[6pt]
0,
& \left\lfloor \dfrac{n-1}{9} \right\rfloor + 2
\le i \le
\left\lfloor \dfrac{2(n-1)}{3} \right\rfloor + 1,
\\[6pt]
1-\left(\dfrac{2(n-1)}{3}-\left\lfloor \dfrac{2(n-1)}{3}\right\rfloor\right),
& i = \left\lfloor \dfrac{2(n-1)}{3} \right\rfloor + 2,
\\[6pt]
1,
& i \ge \left\lfloor \dfrac{2(n-1)}{3} \right\rfloor + 3.
\end{cases}
\end{align*}

\paragraph{Even $n$:}

Let
\[
L=\left\lceil \frac{n-1}{9} \right\rceil,
\qquad
H=\left\lfloor \frac{2(n-1)}{3} \right\rfloor+2.
\]

\begin{align*}
\sigma_1^{\mathrm{init}}(\text{bet} \mid i)
&=
\begin{cases}
1, & i < L,
\\[6pt]
p_L, & i = L,
\\[6pt]
0, & L < i < H,
\\[6pt]
p_H, & i = H,
\\[6pt]
1, & i > H,
\end{cases}
\end{align*}

where
\[
p_H=
\begin{cases}
\dfrac{3}{4}, & n \equiv 4 \pmod{6},\\[6pt]
\dfrac{1}{4}, & \text{otherwise},
\end{cases}
\]

and
\[
p_L=
\begin{cases}
\frac{3}{4}, & n \equiv 0 \pmod{18},\\
\frac{1}{12}, & n \equiv 2 \pmod{18},\\
\frac{1}{4}, & n \equiv 4 \pmod{18},\\
\frac{5}{12}, & n \equiv 6 \pmod{18},\\
\frac{3}{4}, & n \equiv 8 \pmod{18},\\
\frac{11}{12}, & n \equiv 10 \pmod{18},\\
\frac{1}{12}, & n \equiv 12 \pmod{18},\\
\frac{5}{12}, & n \equiv 14 \pmod{18},\\
\frac{7}{12}, & n \equiv 16 \pmod{18}.
\end{cases}
\]

\subsection*{Player 1 after check-bet}

\paragraph{Odd $n$:}

Let
\[
j^*=\left\lfloor \frac{n-1}{6}\right\rfloor+1,
\qquad
q(i)=1-\sigma_1^{\mathrm{init}}(\text{bet}\mid i).
\]

Define
\[
F=\frac{1}{3}\sum_{k=j^*+1}^{n} q(k),
\qquad
c=\min\left\{i>j^* : \sum_{k=j^*+1}^{i} q(k) \ge F \right\}.
\]

\begin{align*}
\sigma_1^{\mathrm{cb}}(\text{call} \mid i)
=
\begin{cases}
0, & i < c,\\[8pt]
1 - \dfrac{F - \sum_{k=j^*+1}^{c-1} q(k)}{q(c)}, & i = c,\\[12pt]
1, & i > c.
\end{cases}
\end{align*}

\paragraph{Even $n$:}

Let
\[
C=\left\lceil \frac{n+2}{3} \right\rceil.
\]

\begin{align*}
\sigma_1^{\mathrm{cb}}(\text{call} \mid i)
&=
\begin{cases}
0, & i < C,
\\[6pt]
\alpha, & i = C,
\\[6pt]
1, & i > C,
\end{cases}
\end{align*}

where
\[
\alpha=
\begin{cases}
\dfrac{1}{4}, & n \equiv 4 \pmod{6},\\[6pt]
\dfrac{3}{4}, & \text{otherwise}.
\end{cases}
\]

\section{Weakest Link}
\label{ap:weakest-link}
In the Weakest Link game show, eight contestants answer a series of trivia questions to accumulate a ``bank'' of money, with one contestant (the ``weakest link'') voted off at each round. When there are 2 contestants remaining, they face off for a series of 5 questions each, with the winner receiving the entire amount that was banked. In theory the champion could win ``up to a million dollars,'' but in practice the total bank ends up being in the $40k-$80k range.

For the first several rounds, it makes a lot of sense to vote for players who are actually the ``weakest,'' since they will be less likely to answer questions correctly and contribute to increasing the amount in the bank throughout the game. But in the final voting round (when 3 contestants remain), it becomes pretty clear that you should actually vote off the ``strongest'' player so that you can go up against a weaker opponent in the final round.

However, this analysis for the final voting round makes several assumptions, which may not hold in practice. First, it assumes that it is clear to you, and to the other players, who the strongest player remaining actually is (and if it is you, then you are screwed). This may be difficult to assess over the relatively small sample of questions each player is given during the game. For example, player A may have correctly answered more questions than player B, but A might have also received easier questions. Furthermore, while it may be evident to you that player A is the strongest contestant, it may not be obvious to player B. If B incorrectly perceives you to be stronger than A and votes for you, while A votes for B, then it is clearly optimal for you to vote for B despite the fact that B is weaker than A.

A second issue is that, regardless of whether the opponents’ perceptions of abilities are correct, they may not understand that they actually want to eliminate the strongest player as opposed to the weakest player in the final voting round. While it seems obvious that one wants to go head-to-head against the weaker remaining opponent, often I see players voting for the clearly weaker remaining opponent. And in the interview of the final contestant eliminated, often their explanation makes it clear that they are not aware of the fact that players would prefer to vote off the strongest player in the last round.

Considering these additional factors, it may actually be optimal under certain circumstances to vote for the weaker remaining contestant, as opposed to the “obviously optimal strategy” of voting for the strongest one (I already gave one example above).

To construct our model, suppose that the total amount of money in the bank to be awarded to the winner is $W>0$ (the loser gets \$0). Suppose that if you are head-to-head against opponent 1 you will win with probability $p_1$, and against opponent 2 you will win with probability $p_2$. Assume that player 2 is stronger than player 1, so that $p_1 > p_2$. Finally, assume that player 1 will vote for you with probability $y_1$ (and therefore will vote for player 2 with probability 1-$y_1$), and that player 2 will vote for you with probability $y_2$. We will assume that obviously no player will vote for themselves.

If there is a three-way tie, we will assume that each player is voted out with probability 1/3. (In reality, the ``statistically strongest link'' from the previous round gets to cast a tie-breaking vote, which is obviously very relevant, but for simplicity I will ignore this aspect of the problem to simplify the analysis). So under this assumption, our expected payoff in the case of a tie equals:
$$\frac{1}{3} (W * p_1) + \frac{1}{3} (W * p_2) + \frac{1}{3} * 0 = \frac{W(p_1 + p_2)}{3}$$

Given this model and assumptions, we now compute your optimal voting strategy. Observe that if both players vote for you, then your vote is irrelevant, since you will be eliminated regardless. So the only relevant cases to consider are when P1 votes for you and P2 votes for P1, and when P2 votes for you and P1 votes for P2.

\begin{enumerate}
\item P1 votes for you, P2 votes for P1: \\
If you vote for P1, you will go head-to-head against P2 and obtain expected payoff $p_2 W$.

If you vote for P2, it will be a three-way tie and you will obtain expected payoff $\frac{W(p_1+p_2)}{3}$, which was calculated above.

\item P2 votes for you, P1 votes for P2: \\
If you vote for P2, you will go head-to-head against P1 and obtain expected payoff $p_1 W$.

If you vote for P1, it will be a tie and you obtain $\frac{W(p_1+p_2)}{3}$.
\end{enumerate}

Assuming we are in either Case 1 or Case 2 (since the other cases are irrelevant, as showed above), the probability that we are in case 1 is $y_1(1-y_2)$, and the probability we are in case 2 is $y_2(1-y_1)$. We need to normalize these so they sum to 1, so the real probability we are in case 1 is:
$$\frac{y_1(1-y_2)}{y_1(1-y_2) + y_2(1-y_1)}$$

And the probability we are in case 2 is:
$$\frac{y_2(1-y_1)}{y_1(1-y_2) + y_2(1-y_1)}$$

Putting this all together, our expected payoff of voting for player 1 is:

$$\frac{y_1(1-y_2)}{y_1(1-y_2) + y_2(1-y_1)} * (W * p_2) + \frac{y_2(1-y_1)}{y_1(1-y_2) + y_2(1-y_1)} * \frac{W(p_1 + p_2)}{3}$$
$$= \frac{y_1(1-y_2) * (W * p_2) + y_2(1-y_1) * \frac{W(p_1 + p_2)}{3}}{y_1(1-y_2) + y_2(1-y_1)}$$

Similarly, our expected payoff of voting for player 2 is:

$$\frac{y_1(1-y_2)}{y_1(1-y_2) + y_2(1-y_1)} * \frac{W(p_1 + p_2)}{3} + \frac{y_2(1-y_1)}{y_1(1-y_2) + y_2(1-y_1)} * (W * p_2)$$
$$= \frac{y_1(1-y_2) * \frac{W(p_1 + p_2)}{3} + y_2(1-y_1) * (W * p_1)}{y_1(1-y_2) + y_2(1-y_1)}$$

So we should vote for player 1 if

$$\frac{y_1(1-y_2) * (W * p_2) + y_2(1-y_1) * \frac{W(p_1 + p_2)}{3}}{y_1(1-y_2) + y_2(1-y_1)}
\geq \frac{y_1(1-y_2) * \frac{W(p_1 + p_2)}{3} + y_2(1-y_1) * (W * p_1)}{y_1(1-y_2) + y_2(1-y_1)}$$

We can multiply both sides by the denominator to eliminate it and obtain an equivalent condition:
$$y_1(1-y_2) * (W * p_2) + y_2(1-y_1) * \frac{W(p_1 + p_2)}{3} 
\geq y_1(1-y_2) * \frac{W(p_1 + p_2)}{3} + y_2(1-y_1) * (W * p_1)$$

If we multiply through and expand both sides, we obtain:
$$y_1Wp_2-y_1y_2Wp_2 + \frac{y_2Wp_1}{3} + \frac{y_2Wp_2}{3}-\frac{y_1y_2Wp_1}{3}-\frac{y_1y_2Wp_2}{3}$$
$$\geq \frac{y_1Wp_1}{3}+\frac{y_1Wp_2}{3}-\frac{y_1y_2Wp_1}{3}-\frac{y_1y_2Wp_2}{3}+y_2Wp_1-y_1y_2Wp_1$$

We can simplify this to obtain:
$$y_1Wp_2 - y_1y_2Wp_2 + \frac{y_2Wp_1}{3} + \frac{y_2Wp_2}{3} \geq \frac{y_1Wp_1}{3} + \frac{y_1Wp_2}{3} + y_2Wp_1 - y_1y_2Wp_1$$

Multiplying both sides by 3:
$$3y_1Wp_2 - 3y_1y_2Wp_2 + y_2Wp_1 + y_2Wp_2 \geq y_1Wp_1 + y_1Wp_2 + 3y_2Wp_1 -3y_1y_2Wp_1$$

Simplifying further:
$$2y_1Wp_2 + y_2Wp_2 + 3y_1y_2Wp_1 \geq 2y_2Wp_1 + y_1Wp_1 + 3y_1y_2Wp_2$$

Dividing all terms by $W$, we see that we should vote for player 1 iff:
\begin{equation}
2y_1p_2 + y_2p_2 + 3y_1y_2p_1 \geq 2y_2p_1 + y_1p_1 + 3y_1y_2p_2
\end{equation}

One immediate observation is that the optimal strategy does not depend on W.

We can further show that if $p_2 \leq \frac{p_1}{2}$, then it is always optimal to vote for player 2 (the stronger player) regardless of the beliefs of the strategies taken by the other players. For brevity we omit the proof of this result.

\end{document}